\documentclass[10pt,letterpaper]{article}
\usepackage[top=0.85in,left=2.75in,footskip=0.75in]{geometry}

\usepackage{amsmath,amssymb}

\usepackage{bbm,hyperref,latexsym,alltt}
\usepackage{algorithm}
\usepackage{changepage}

\usepackage[utf8x]{inputenc}

\usepackage{textcomp,marvosym}

\usepackage{cite}

\usepackage{nameref,hyperref}

\usepackage{microtype}
\DisableLigatures[f]{encoding = *, family = * }

\usepackage[table]{xcolor}

\usepackage{array}

\newcolumntype{+}{!{\vrule width 2pt}}

\newlength\savedwidth

\raggedright
\usepackage[aboveskip=1pt,labelfont=bf,labelsep=period,justification=raggedright,singlelinecheck=off]{caption}

\makeatletter
\renewcommand{\@biblabel}[1]{\quad#1.}
\makeatother

\usepackage{lastpage,fancyhdr,graphicx}
\usepackage{epstopdf}
\fancyheadoffset[L]{2.25in}
\fancyfootoffset[L]{2.25in}
\newcommand{\tmem}[1]{{\em #1\/}}

\newcommand{\tmmathbf}[1]{\ensuremath{\boldsymbol{#1}}}
\newcommand{\tmop}[1]{\ensuremath{\operatorname{#1}}}
\newcommand{\tmstrong}[1]{\textbf{#1}}
\newcommand{\tmtextbf}[1]{{\bfseries{#1}}}

\newcommand{\tmtexttt}[1]{{\ttfamily{#1}}}
\newcommand{\tmverbatim}[1]{{\ttfamily{#1}}}
\newenvironment{itemizeminus}{\begin{itemize} }{\end{itemize}}
\newenvironment{proof}{\noindent\textbf{Proof\ }}{\hspace*{\fill}$\Box$\medskip}
\newenvironment{tmcode}[1][]{\begin{alltt} }{\end{alltt}}
\newenvironment{tmindent}{\begin{tmparmod}{1.5em}{0pt}{0pt} }{\end{tmparmod}}
\newenvironment{tmparmod}[3]{\begin{list}{}{\setlength{\topsep}{0pt}\setlength{\leftmargin}{#1}\setlength{\rightmargin}{#2}\setlength{\parindent}{#3}\setlength{\listparindent}{\parindent}\setlength{\itemindent}{\parindent}\setlength{\parsep}{\parskip}} \item[]}{\end{list}}

\newtheorem{proposition}{Proposition}

\begin{document}
\vspace*{0.2in}

\begin{flushleft}
{\Large
\textbf\newline{SymDPoly: symmetry-adapted moment relaxations for noncommutative
polynomial optimization} % Please use "sentence case" for title and headings (capitalize only the first word in a title (or heading), the first word in a subtitle (or subheading), and any proper nouns).
}
\newline
\\
Rosset Denis\textsuperscript{1*}
\\
\bigskip
\textbf{1} Perimeter Institute for Theoretical Physics, Waterloo,
  Ontario, Canada, N2L 2Y5
\\
\bigskip

% Use the asterisk to denote corresponding authorship and provide email address in note below.
* physics@denisrosset.com

\end{flushleft}

\section*{Abstract}
Semidefinite relaxations are widely used to compute upper bounds on the
objective of optimization problems involving noncommutative polynomials.
Such optimization problems are prevalent in quantum information. We present
an algorithm able to discover automatically and exploit the symmetries
present in the problem formulation. We also provide an open source software
library written in Scala
(\href{https://github.com/denisrosset/symdpoly}{https://github.com/denisrosset/symdpoly})
that computes symmetry-adapted semidefinite relaxations with interfaces to a
variety of open-source and commercial semidefinite solvers. We discuss the
advantages of symmetrization, namely reductions in memory use, computation
time, and increase in the solution precision.

% \linenumbers AUTHOR
\section*{Introduction}

Semidefinite programming is a prevalent tool to study quantum systems. As
density matrices are semidefinite matrices of trace one, semidefinite programs
naturally address questions related to unambiguous state
discrimination~{\cite{Eldar2003}}, entanglement
detection~{\cite{Doherty2004}}, entanglement
measures~{\cite{Rains2001,Vidal2002}}, measurement
incompatibility~{\cite{Wolf2009}} and steering~{\cite{Cavalcanti2017b}}, among
other applications. Polynomial optimization problems are also frequent in quantum information.
Systems of commutative polynomial equations appear in the characterization of
sets of local correlations; for example in the study of
network-locality~{\cite{Rosset2017a}}, the study of causal
structures~{\cite{Lee2015}}, and maximal violations of Bell inequalities for
given states~{\cite{Liang2007a}}. These polynomial problems can be handled by
a hierarchy of semidefinite relaxations based on sums of squares
formulations~{\cite{Ferrier1998,Nesterov2000,Lasserre2001,Parrilo2003}}.
Other questions are neatly formulated as optimizations over noncommutative
polynomial rings of operators. In this second setting, moment relaxations are
used to characterize the quantum set of
correlations~{\cite{Navascues2007,Doherty2008,Navascues2008a}}, provide
dimensional bounds~{\cite{Navascues2014,Navascues2015a}}, quantify
entanglement~{\cite{Moroder2013}} or characterize in a robust manner quantum
devices~{\cite{Yang2014,Bancal2015}}. Similar hierarchies were studied in
their mathematical abstract
setting~{\cite{Helton2002,Cafuta2011,Navascues2012,Burgdorf2012,Burgdorf2013}}.
The complexity of those semidefinite relaxations increases rapidly with the
relaxation degree. To address that problem, the symmetries of the original
problem can be applied to the semidefinite relaxations and reduce the problem
size. The technique was introduced in~{\cite{Gatermann2004}} in the
commutative case and reviewed in~{\cite{Parrilo2005,Riener2012}}.

In quantum information, symmetry techniques have been applied to semidefinite
programs: in quantum control~{\cite{Branczyk2007}} or in quantum
metrology~{\cite{Chiribella2012}}. In the specific case of sums of squares
relaxations, symmetry techniques were applied to
self-testing~{\cite{Bamps2015}} and translation-invariant Bell
inequalities~{\cite{Fadel2017}}. We also mention the related work in
preparation~{\cite{Tavakoli2018a}} applying to optimization over
finite-dimensional quantum systems, based on randomized sampling
rather than exact algebraic methods.
As the problem sizes grow, semidefinite relaxations are not written by hand
but rather constructed using software libraries. Among others, we mention the
libraries YALMIP~{\cite{Lofberg2009a}}, GloptiPoly~{\cite{Henrion2009}},
SOSTOOLS~{\cite{sostools}}, SparsePOP~{\cite{Waki2008}} for the commutative
case; NCSOStools~{\cite{Cafuta2011}} and Ncpol2Sdpa~{\cite{Wittek2015}} in the
noncommutative case. Our implementation is particularly influenced by this
last package.

In the present manuscript, we introduce symmetry-adapted moment relaxations
for a variety of noncommutative optimization problems arising from quantum
information scenarios, along with a software library that automates their
formulation.
The use of symmetries leads to huge efficiency gains in that context.
Consider a semidefinite program in the canonical
form that involves a matrix of size $n \times n$ on a space of affine
dimension $m$. When using a primal-dual barrier method such as implemented by
SDPA~{\cite{Yamashita2012}}, the memory requirements scale as $\mathcal{O}
(m^2 + mn^2)$, and the CPU time per iteration scales\footnote{In the memory
requirements, the $\mathcal{O} (m^2)$ term represents the Schur complement
matrix and $\mathcal{O} (m n^2)$ is an upper bound that depends on the matrix
sparsity. The scaling of the CPU time per iteration has three parts: the
computation of the Schur complement in $\mathcal{O} (m n^3 + m^2 n^2)$, the
Cholesky decomposition in $\mathcal{O} (m^3)$ and various other contributions
in $\mathcal{O} (n^3)$.} as in $\mathcal{O} (m^3 + n^3 + m n^3 + m^2 n^2)$.
Thus, any method that reduces $m$ and/or $n$ has a great impact on the memory
and CPU requirements. Moreover, some recent SDP solvers only converge when the
solution has no degeneracies~{\cite{Zhao2010}}, and in general, reducing the
complexity can improve the precision of the solutions by 1-2 orders of
magnitude, as we will see in the present manuscript.

The manuscript is divided in four parts. In
Section~\ref{Sec:OptimizationOver}, we define formally optimization problems
over noncommutative polynomials and their symmetries. In contrast to previous
presentations, we emphasize the use of rewriting rules during monomial
expansion. In Section~\ref{Sec:MomentRelaxations}, we review the semidefinite
hierarchies based on moment relaxations; most importantly, we express the
variants due to Moroder et al.~{\cite{Moroder2013}} and Burgdorf et
al.~{\cite{Cafuta2011,Burgdorf2012,Burgdorf2013}} in a common framework. In
Section~\ref{Sec:Implementation}, we discuss the choices made in our
implementation, including the algorithms enabling symmetric formulations. We
present a practical application in Section~\ref{Sec:Application} by computing
high precision bounds for the $I_{3322}$
inequality~{\cite{Froissart1981,Sliwa2003,Collins2004}}.

\section{Optimization over noncommutative polynomials}

\label{Sec:OptimizationOver}We assume that the reader is familiar with moment
relaxations, as introduced
in~{\cite{Navascues2007,Doherty2008,Navascues2008a}} for problems in quantum
information. Our symmetric moment relaxations apply to optimization problems
defined using noncommutative polynomials. We use a modified version of the
presentation~{\cite{Navascues2012}}: first, we define the monomials
involved, their rewriting rules and symmetries, before defining noncommutative
polynomials over those monomials in a second step, and finally express
optimization problems over those polynomials. We are not overly concerned by
technicalities such as proving convergence: the optimal values and convergence
properties of our hierarchies match the original formulations published in the
literature.

\subsection{Monomials}

We consider a set of {\tmem{letters}} $\{ x_1, x_2, \ldots, x_n \}$, along
with an {\tmem{involution}} $\ast$ such that $((x_1)^{\ast})^{\ast} = x_1$. We
collect \ these letters in the set $\tmmathbf{x}$ along with their images
under $\ast$
\[ \tmmathbf{x}= \{ x_1, x_2, \ldots, x_n, x_1^{\ast}, x^{\ast}_2, \ldots,
   x_n^{\ast} \} . \]
We write $\mathcal{S}_{\ast}$ the group of all permutations of elements
$\tmmathbf{x}$ that commute with the involution: we require for all $\pi \in
\mathcal{S}_{\ast}$ that $\pi (x_i^{\ast}) = \pi (x_i)^{\ast}$. We write
$\mathcal{W}$ the free monoid on $\tmmathbf{x}$, defined as follows. A
{\tmem{word}} or {\tmem{monomial}} $w \in \mathcal{W}$ is written
\[ w = w_1 w_2 \ldots w_m \]
with $m = | w |$ the {\tmem{length}} of $w$. The identity element is the empty
word of zero length, denoted by $1$, and the monoid operation $\cdot$ is word
concatenation, i.e. for words $v, w \in \mathcal{W}$ written over letters as
$v = v_1 \ldots v_{\ell}$ and $w = w_1 \ldots w_m$ we have
\[ v \cdot w = v_1 \ldots v_{\ell} w_1 \ldots w_m . \]
The involution $\ast$ acts on a word $w \in \mathcal{W}$ as
\[ w^{\ast} = w^{\ast}_m w^{\ast}_{m - 1} \ldots w^{\ast}_1, \]
so that $(v \cdot w)^{\ast} = w^{\ast} \cdot v^{\ast}$. Thus $\mathcal{W}$ is
a $\ast$-monoid~{\cite{Drazin1978}}. An element $\pi \in \mathcal{S}_{\ast}$
acts on $w = w_1 w_2 \ldots w_m \in \mathcal{W}$ as
\[ \pi (w) = \pi (w_1) \pi (w_2) \ldots \pi (w_m), \qquad \pi (1) = 1. \]
We extend our free monoid $\mathcal{W}$ to $\mathcal{W}^0 =\mathcal{W} \cup \{
0 \}$ by the addition of a zero element. We define formally
\[ 0^{\ast} = 0, \qquad 0 \cdot 0 = 0, \qquad 0 \cdot w = 0, \qquad w \cdot 0
   = 0, \qquad \text{for all } w \in \mathcal{W}, \]
\[ \pi (0) = 0 \quad \text{for all } \pi \in \mathcal{S}_{\ast}, \qquad
   \text{and} \qquad | 0 | = - \infty . \]
A {\tmem{congruence}} $\sim$ on $\mathcal{W}^0$ is an equivalence relation
that satisfies, for all $x, y, a, b \in \mathcal{W}^0$,
\[ \left( v \sim x \text{ and } w \sim y \right) \qquad \Rightarrow \qquad
   v^{\ast} \sim x^{\ast} \text{ and } v \cdot w \sim x \cdot y. \]
Given a word $w \in \mathcal{W}^0$, we write $[w]_{\sim} = \{ v \in
\mathcal{W}^0 : v \sim w \}$ its congruence class, and $\tilde{\mathcal{W}}
=\mathcal{W}^0 / \sim$ the set of all such congruence classes. We define a
binary operation $\cdot$ on the set $\tilde{\mathcal{W}}$ by
\[ [v]_{\sim} \cdot [w]_{\sim} = [v \cdot w]_{\sim} \qquad \text{for all } v,
   w \in \mathcal{W}^0, \]
and we easily verify that $\tilde{\mathcal{W}}$ is a $\ast$-monoid, the
{\tmem{quotient monoid}} of $\mathcal{W}$ by $\sim$. We define the symmetry
group $\mathcal{S}_{\sim} \subseteq \mathcal{S}_{\ast}$ as containing
permutations that preserve congruence
\begin{equation}
  \label{Eq:GroupCongruenceCompatible} \mathcal{S}_{\sim} = \left\{ \pi \in
  \mathcal{S}_{\ast} : \pi (v) \sim \pi (w) \text{ for all } v, w \in
  \mathcal{W}^0 \text{ such that } v \sim w \right\} .
\end{equation}
Then the action of $\mathcal{S}_{\sim}$ on $\tilde{\mathcal{W}}$ is well
defined:
\[ \pi ([w]_{\sim}) = [\pi (w)]_{\sim} \qquad \text{for all } \pi \in
   \mathcal{S}_{\sim} \text{ and } w \in \mathcal{W}^0 . \]
For computational purposes, the congruence $\sim$ is represented by a set of
rewriting rules $R = \{ v_1 \rightarrow w_1, v_2 \rightarrow w_2, \ldots \}$,
which, given a monomial $x v_i y$, applies as
\[ x v_i y \rightarrow x w_i y \qquad \text{for all } x, y \in \mathcal{W}^0,
   (v_i \rightarrow w_i) \in R. \]
For example, these rewriting rules can encode commutation relations ($x_j x_i
\rightarrow x_i x_j$ for some pairs $i, j$).

A word $u$ is in {\tmem{normal form}} if it cannot be rewritten any further.
We write $\mathcal{N}_R (u)$ the normal form obtained after repeated
application of the rewriting rules $R$; we require the rewriting system to be
{\tmem{confluent}}, which means that the normal form of $u$ does not depend on
the order of rule application. We then define formally the congruence $\sim$
from the rewriting system $R$:
\[ v \sim w \qquad \Leftrightarrow \qquad \mathcal{N}_R (v) =\mathcal{N}_R (w)
   . \]
The definition~(\ref{Eq:GroupCongruenceCompatible}) becomes
\[ \mathcal{S}_{\sim} = \left\{ \pi \in \mathcal{S}_{\ast} : \mathcal{N}_R
   (\pi (v)) =\mathcal{N}_R (\pi (w)) \text{ for all } v, w \in \mathcal{W}^0
   \text{ such that } \mathcal{N}_R (v) =\mathcal{N}_R (w) \right\} . \]
Confluent rewriting systems can be constructed and verified using the
Knuth-Bendix completion algorithm, whose description and implementation is
outside the scope of our work. We require the user of our software to provide
a confluent rewriting system (confluent rules for common correlations
scenarios are provided below).

We work with rewriting rules that such that $\mathcal{N}_R (w)$ has minimal
length over the equivalence class of $w$. Thus, we define the length of $[w]$
as the length of $\mathcal{N}_R (w)$.

\subsection{Rewriting rules for quantum correlation scenarios}

We now give two examples of monoids $\tilde{\mathcal{W}}$ along with their
rewriting rules $R$ and symmetry group $\mathcal{S}_{\sim}$.

\paragraph*{Binary outputs. ---}Consider a two-party Bell scenario where Alice
(resp. Bob) has input $x = 1,\ldots,m$ taking $m$ distinct values (respectively $y = 1,
\ldots, m$) \ and binary outputs $a = \pm 1$ (resp. $b = \pm 1$). We write $A_x$
(resp. $B_y$) the formal variable associated with the projective measurements
of Alice with eigenvalues $- 1$ and $+ 1$ (and the same for Bob). We have
\begin{equation}
  \label{Eq:BinaryOutcomesVariables} \tmmathbf{x}= \{ A_1, \ldots, A_m, B_1,
  \ldots, B_m, A^{\ast}_1, \ldots, A^{\ast}_m, B^{\ast}_1, \ldots, B^{\ast}_m
  \}
\end{equation}
and equivalence of monomials is defined by the rewriting rules
\[ \label{Eq:BinaryOutcomesRules} R = \left\{ A_x^{\ast} \rightarrow A_x,
   \qquad B_y^{\ast} \rightarrow B_y, \qquad B_y A_x \rightarrow A_x B_y,
   \qquad A_x A_x \rightarrow 1, \qquad B_y B_y \rightarrow 1 \right\} \]
for $x, y = 1, \ldots, m$. To simplify the computations in this self-adjoint case, we identify
$A_x^{\ast} = A_x$ and $B_y^{\ast} = B_y$ directly in $\tmmathbf{x}$. The
symmetry group $\mathcal{S}^{\tmmathbf{x}}_{\sim}$ contains the permutations
preserving the partition $\{ \{ A_1, \ldots, A_m \}, \{ B_1, \ldots, B_m \}
\}$.

\paragraph*{Multiple outputs. ---}We now generalize this example to the case
of $d \geqslant 2$ outcomes. Let Alice (respectively Bob) choose between $m$
projective measurements, each with $d$ outcomes. We write $A_{a | x}$
the formal variable associated with the projector corresponding to the output
$a$ and input $x$ (respectively $B_{b | y}$ for Bob). As projectors are
Hermitian, we identify $A^{\ast}_{a | x} = A_{a | x}$ and $B_{b
| y}^{\ast} = B_{b | y}$ and have
\[ \tmmathbf{x}= \{ A_{a | x} \}_{a, x} \cup \{ B_{b | y} \}_{b,
   y} \]
with the rewrite rules
\begin{multline*}
 R = \big\{ A_{a | x} A_{a | x} \rightarrow A_{a | x},
   \quad A_{a | x} A_{a' | x} \rightarrow 0, \quad B_{b |
     y} B_{b| y} \rightarrow B_{b | y}, \\
   \quad B_{b | y}
   B_{b' | y} \rightarrow 0, \quad B_{b | y} A_{a | x}
   \rightarrow A_{a | x} B_{b | y} \big\}
 \end{multline*}
for $a, a', b, b' = 1, \ldots, d$ and $x, y = 1, \ldots, m$, with $a' \neq a$,
$b' \neq b$. The symmetry group $\mathcal{S}_{\sim}$ contains all permutations
preserving the partitions $P_1$ and $P_2$:
\[ P_1 = \{ \{ A_{1 | 1}, \ldots, A_{d | m} \}, \{ B_{1 |
   1}, \ldots, B_{d | m} \} \}, \]
\[ P_2 = \{ \{ A_{1 | 1}, \ldots, A_{d | 1} \}, \ldots, \{ A_{1
   | m}, \ldots, A_{d | m} \}, \{ B_{1 | 1}, \ldots, B_{d
   | 1} \}, \ldots, \{ B_{1 | m}, \ldots, B_{d | m} \} \}
   . \]
Note that the relation $\sum_a A_{a | x} = \sum_b B_{b | y} = 1$
is not captured at the level of monomials.

\subsection{Noncommutative polynomials}

Given a set of letters $\tmmathbf{x}$, a list of rewrite rules $R$, and a
field $\mathbbm{K} \in \{ \mathbbm{R}, \mathbbm{C} \}$, we write $\mathbbm{K}
[\tilde{\mathcal{W}}] = \{ p \}$ the set of formal sums of the form
\[ p = \sum_{[w] \in \tilde{\mathcal{W}}} p_{[w]} [w], \qquad p_{[w]} \in
   \mathbbm{K}, \]
where $p_{[w]} = 0$ for all but finitely many $w$; in addition, we require
that $p_{[0]} = 0$. Such formal sums are noncommutative polynomials over the
monomials $\tilde{W}$. The {\tmem{degree}} of $p$ is given by the largest
$[w]$-length corresponding to a nonzero $p_{[w]}$. Addition on $\mathbbm{K}
[\tilde{\mathcal{W}}]$ is defined component-wise, and the multiplication is
defined by having the elements of $\mathbbm{K}$ commute with the elements of
$\tilde{\mathcal{W}}$ (and forcing $p_{[0]} = 0$). The involution acts as
\[ p^{\ast} = \sum_{[w] \in \tilde{\mathcal{W}}} (p_{[w]})^{\ast} [w]^{\ast},
\]
where $k^{\ast}$ is the complex conjugate of $k \in \mathbbm{K}$. With these
definitions, $\mathbbm{K} [\tilde{\mathcal{W}}]$ is a $\ast$-algebra. The
permutation group $\mathcal{S}_{\sim}$ acts naturally on $p$:
\[ \pi (p) = \sum_{[w] \in \tilde{\mathcal{W}}} p_{[w]}  [\pi (w)] \quad
   \text{for } \pi \in \mathcal{S}_{\sim} . \]

\subsection{Optimization problems}

Consider now the set $\mathcal{B} (\mathcal{H})$ of bounded operators on a
Hilbert space $\mathcal{H}$ defined on the field $\mathbbm{K}$, with
$\mathbbm{1} \in \mathcal{B} (\mathcal{H})$ the identity operator. Given a set
of operators $\tmmathbf{X}= (X_1, \ldots, X_n)$ and a polynomial $p \in
\mathbbm{K} [\tilde{\mathcal{W}}]$, we define the operator $p (X) \in
\mathcal{B} (\mathcal{H})$ by replacing
\[ 1 \rightarrow \mathbbm{1}, \qquad x_i \rightarrow X_i, \qquad x_i^{\ast}
   \rightarrow X_i^{\ast}, \]
in the normal form of $p$, where $X_i^{\ast}$ is the adjoint of $X_i$. Note
that the substitution is consistent only if the operators $\tmmathbf{X}$
satisfy the same relations $R$ as the variables $\tmmathbf{x}$. We evaluate
those polynomial on vectors $\phi \in \mathcal{H}$ by
\[ \langle p (\tmmathbf{X}) \rangle_{\phi} = \langle \phi | p (\tmmathbf{X}) |
   \phi \rangle, \]
noting that other choices are possible (see
Section~\ref{Sec:NPAGeneralization}). A polynomial for which $p = p^{\ast}$
is {\tmem{Hermitian}}, and in that case $p (\tmmathbf{X}) = p^{\ast}
(\tmmathbf{X})$ is a Hermitian operator. For Hermitian $p = p^{\ast}$, the
quantity $\langle p (\tmmathbf{X}) \rangle_{\phi}$ is real; this motivates the
following canonical form of a optimization problem over noncommutative
polynomials.
\begin{equation}
  \label{Eq:OptimizationProblem} p^{\star} = \sup_{\tmmathbf{X}, \phi} \langle
  p (\tmmathbf{X}) \rangle_{\phi}
\end{equation}
\begin{equation}
  \label{Eq:Constraints} \begin{array}{lrl}
    \text{subject to } & \langle \mathbbm{1} \rangle_{\phi} = 1, & \\
    & q_i (\tmmathbf{X}) \succeq 0, & i \in \mathcal{I},\\
    & r_j (\tmmathbf{X}) = 0, & j \in \mathcal{J},\\
    & \langle s_k (\tmmathbf{X}) \rangle_{\phi} \geq 0, & k \in \mathcal{K},
  \end{array}
\end{equation}
where $\mathcal{I}$, $\mathcal{J}$ and $\mathcal{K}$ are index sets, all
$q_i$, $s_k$ and $p$ are Hermitian, the optimization is carried out over all
states $\phi \in \mathcal{H}$ defined on Hilbert spaces $\mathcal{H}$ of
arbitrary dimension and operators $\tmmathbf{X}= (X_1, \ldots, X_n)$ in
$\mathcal{B} (\mathcal{H})$ that satisfy the rewrite rules $R$. We denote by
$q_i (\tmmathbf{X}) \succeq 0$ the positive semidefiniteness of $q_i
(\tmmathbf{X})$, i.e. $\langle \psi | q_i (\tmmathbf{X}) | \psi \rangle \geq
0$ for all $\psi \in \mathcal{H}$ (not only for $\psi = \phi$).

Similarly to~{\cite{Wittek2015}} where they are called binomials, we allow
efficient handling of two-term equalities $v - w = 0$, where $v, w \in
\mathcal{W}$, by handling them at the level of the congruence $\sim$.

\subsection{Symmetries of optimization problems}

While the group $\mathcal{S}_{\sim}$ preserved the structure of the congruence
$\sim$, we define the {\tmem{ambient group}} $G \subseteq S_{\sim}$ that
preserves feasibility under the constraints~(\ref{Eq:Constraints})
\[ G = \left\{ \quad g \in S_{\sim} \quad : \quad (\tmmathbf{X}, \phi) \text{
   is feasible} \quad \Rightarrow \quad (g (\tmmathbf{X}), \phi) \text{ is
   feasible} \quad  \right\} . \]
The {\tmem{symmetry group}} $G_{\star}$ of the optimization problem also
preserves optimality:
\[ G_{\star} = \left\{ \quad g \in G \quad : \quad \langle p (g
   (\tmmathbf{X})) \rangle_{\phi} = \langle p (\tmmathbf{X}) \rangle_{\phi}
   \quad \text{for all} \quad (\tmmathbf{X}, \phi) \text{ is feasible} \quad
   \right\} . \]

\subsection{Signed monomials and generalized permutations}

For efficiency, we generalize slightly the permutations used to build the
groups $\mathcal{S}_{\ast}$, $\mathcal{S}_{\sim}$, $G$ and $G^{\star}$. A
{\tmem{generalized permutation}} $\pi$ on $n$ elements is defined by the
sequence of images
\[ (\pi_1, \ldots, \pi_n) = (\pm \rho_1, \ldots, \pm \rho_n), \]
where $\rho: i \mapsto
\rho_i$ is a standard permutation. The generalized
permutation $\pi$ acts on the integers $\{ - n, \ldots, - 1, 1, \ldots, n \}$
by
\[ \pi (i) = \tmop{sign} (i) \pi_{| i |} . \]
The group of the generalized permutations on $n$ elements is also called the
signed symmetric group. In the present case, we write
$\mathcal{S}_{\ast}^{_{\pm}}$ the signed symmetric group acting on the
{\tmem{signed letters}}
\[ \tmmathbf{x}^{\pm} = \{ \pm x_1, \ldots, \pm x_n, \pm x_1^{\ast}, \ldots,
   \pm x_n^{\ast} \} . \]
The signed monomials $\mathcal{W}^{\pm} = \{ w^{\pm} = \omega w_1 \ldots w_m
\}$ are defined as product of letters preceded with a sign $\omega = \pm 1$.
Given such $w^\pm\in\mathcal{W}^\pm$, we define $\tmop{sign} (w^{\pm}) = \omega$ and $\tmop{abs} (w^{\pm}) = |
w^{\pm} | = w_1 \ldots w_m$.
The action of $\pi \in \mathcal{S}_{\ast}^{_{\pm}}$ on $w^{\pm} \in
\mathcal{W}^{\pm}$ is
\[ \pi (\omega x_{i_1} \ldots x_{i_m}) = \omega \tmop{sign} (\pi (i_1) \ldots
   \pi (i_m)) x_{| \pi (i_1) |} \ldots x_{| \pi (i_m) |} . \]
We consider the equivalence classes of $\mathcal{W}^{\pm}$ under the rewriting
rules $R$, noting that $R$ does not affect the sign: thus, elements of
$\mathcal{W}^{\pm} / \sim$ are simply written $\pm [w]$ with $[w] \in
\mathcal{W}/ \sim$.
Similarly, $\mathcal{S}_{\ast}^{_{\pm}}$ can be restricted to be compatible
with the congruence $\sim$, so that $\mathcal{S}_{\sim}^{\pm}$ acts
consistently on the equivalence classes of $\mathcal{W}^{\pm} / \sim$. For
example, the rewrite rule $x_i x_i \rightarrow x_i$ is not compatible with the
generalized permutation $\pi$ that sends $x_i$ to $\pi (x_i) = - x_i$.
However, $\pi (x_i) = - x_i$ would be compatible with the rewrite rules for
variables~(\ref{Eq:BinaryOutcomesRules}), and generalized permutations lead to
huge gains of efficiency on quantum correlation scenarios involving binary
outputs. Finally, we identify $\mathbbm{K}$-linear combinations of signed
monomials $\mathbbm{K} [\tilde{W}^{\pm}]$ with polynomials in $\mathbbm{K}
[\tilde{W}]$ by writing
\[ p = \sum_{[w^{\pm}] \in \tilde{\mathcal{W}}^{\pm}} p_{[w^{\pm}]} [w^{\pm}]
   = \sum_{[w^{\pm}]} \tmop{sign} ([w^{\pm}]) p_{[w^{\pm}]}  [\tmop{abs}
   (w^{\pm})], \]
and the action of $\mathcal{S}_{\sim}^{\pm}$ on $\mathbbm{K} [\tilde{W}]$
follows.

\subsection{Example: the CHSH inequality}

\label{Sec:CHSHInequality}We consider a two-party Bell scenarios with binary
inputs and outputs, i.e. $x, y = 0, 1$ and $a, b = \pm 1$. The measurements
are represented by Hermitian operators $\tmmathbf{x}= \{ A_0, A_1, B_0, B_1
\}$, along with the rewriting rules~(\ref{Eq:BinaryOutcomesRules}). The group
$S_{\sim}^{\pm}$ is generated by the generalized permutations (abusing slightly the notation)
\begin{equation}
  \label{Eq:CHSHGenerators} \pi_1 = (B_0, B_1, A_0, A_1), \qquad \pi_2 = (A_0,
  A_1, B_1, B_0), \qquad \pi_3 = (A_0, - A_1, B_0, B_1),
\end{equation}
where $\pi_1$ permutes the parties, $\pi_2$ permutes the inputs of Bob, and
$\pi_3$ is a conditional permutation of the outputs of Alice. The group
$S_{\sim}^{\pm}$ is of order 128. We do not need to add explicitly the
constraints
\[ (1 \pm A_x) \succeq 0, \qquad (1 \pm B_y) \succeq 0, \]
as $(1 \pm A_x) / 2$ are both projectors: for example $(1 + A_x) / 2 \sim (1 + A_x)^2 /
4$. Thus $G^{\pm} = S_{\sim}^{\pm}$.

Our goal is to maximize the value of the CHSH expression~{\cite{Clauser1969}}
\[ p_{\text{CHSH}} = [A_0 B_0] + [A_0 B_1] + [A_1 B_0] - [A_1 B_1] \]
without constraints $q_i, r_j$ or $s_k$. The expression $p_{\tmop{CHSH}}$ is
symmetric under the group $G_{\star}^{\pm}$ generated by
\[ \sigma_1 = \pi_1, \qquad \sigma_2 = \pi_2 \pi_3, \qquad \sigma_3 = (- A_0,
   - A_1, - B_0, - B_1) \]
of order $16$.

\section{Moment relaxations}
\label{Sec:MomentRelaxations}

We now define moment relaxations of the optimization problems we just
introduced. First, we present their standard formulation, before discussing
their symmetrization. We conclude this section by solving a concrete example by hand.

\subsection{Definition}
Moment relaxations arise from the existence of a
linear functional $\mathcal{L}: \mathbbm{K} [\tilde{\mathcal{W}}] \rightarrow
\mathbbm{K}$ which satisfies
\begin{equation}
  \label{Eq:LinearFunctional} \mathcal{L} ([1]) = 1, \quad \mathcal{L} (f)
  =\mathcal{L} (f^{\ast})^{\ast}, \quad \mathcal{L} (f^{\ast} f) \geq 0, \quad
  \text{for all } f \in \mathbbm{K} [\tilde{\mathcal{W}}],
\end{equation}
\[ \begin{array}{rll}
     \mathcal{L} (f^{\ast} q_i f) \geq 0, & \text{for all } f \in \mathbbm{K}
     [\tilde{\mathcal{W}}], & i \in \mathcal{I},\\
     \mathcal{L} (fr_j g) = 0, & \text{for all } f, g \in \mathbbm{K}
     [\tilde{\mathcal{W}}], & j \in \mathcal{J},\\
     \mathcal{L} (s_k) \geq 0, &  & k \in \mathcal{K}.
   \end{array} \]
Any feasible solution $(\tmmathbf{X}, \phi)$ defines a linear functional
\[ \mathcal{L}_{(\tmmathbf{X}, \phi)} (f) = \langle f (\tmmathbf{X})
   \rangle_{\phi} \]
that satisfies~(\ref{Eq:LinearFunctional}). Moment relaxations are defined as
a relaxation of the constraints~(\ref{Eq:LinearFunctional}), by considering
test polynomials $f, g \in \mathbbm{K} [\tilde{\mathcal{W}}]$ such that the
final expressions evaluated by $\mathcal{L}$ involve only polynomials of
maximal degree $2 d$, for some $d \geqslant 1$:
\[ \mathbbm{K} [\tilde{\mathcal{W}}]^{2 d} = \{ f \in \mathbbm{K}
   [\tilde{\mathcal{W}}] : \deg (f) \leqslant 2 d \} . \]
Remark that the restriction $\mathcal{L}: \mathbbm{K} [\tilde{\mathcal{W}}]^{2
d} \rightarrow \mathbbm{K}$ is fully defined by
\[ \mathcal{L} (f) = \sum_{\deg [w] \leq 2 d} f_{[w]} y_{[w]}, \qquad y_{[w]}
   \equiv \mathcal{L} ([w]), \]
as by linearity $\mathcal{L}$ is completely characterized by the values
$\vec{y} \in \mathbbm{K}^{N_{2 d}}$, where $N_D$ is the number of $[w]$ of
degree at most $D$.

We are now ready to express our constraints~(\ref{Eq:LinearFunctional}) in
semidefinite form. The linear constraints are:
\begin{equation}
  \label{Eq:SDPConstraints1} y_{[1]} = 1, \quad y_{[w]} =
  (y_{[w^{\ast}]})^{\ast}, \quad \sum_{[u]  [v]  [w]} (r_j)_{[v]} y_{[u v w]}
  = 0, \quad \sum_{[w]} (s_k)_{[w]} y_{[w]} \geq 0,
\end{equation}
while semidefinite constraints are given by the {\tmem{moment matrix}} $\Xi$
and the {\tmem{localizing matrices}} $\Lambda_i$:
\begin{equation}
  \label{Eq:SDPConstraints2} \Xi = \sum_{\deg ([u], [v]) \leq d} y_{[u^{\ast}
  v]} E^{u v} \succeq 0, \qquad \Lambda_i = \sum_{\deg ([u], [v]) \leq d, [w]}
  (q_i)_{[w]} y_{[u^{\ast} w v]} E^{u v} \succeq 0,
\end{equation}
with $E^{u v}$ a $N_d \times N_d$ matrix whose rows and columns indices $(r,
c)$ correspond to an ordering of the monomials $[w]$ and
\[ (E^{u v})_{r, c} = \left\{ \begin{array}{ll}
     1 & \text{if } r = u \text{ and } c = v,\\
     0 & \text{otherwise},
   \end{array} \right. \]
and, above, sums run over $[u], [v], [w]$ with the restriction that the
$y_{[\ldots]}$ is indexed by a monomial of degree at most $2 d$; $i, j, k$ run
over their respective index sets $\mathcal{I}, \mathcal{J}, \mathcal{K}$. Note
that, depending on the degree of $q_i$, rows and columns of the matrices
$\Lambda_i$ are omitted, see~{\cite{Navascues2012}} for details.
The final semidefinite program is given by
\begin{equation}
  \label{Eq:SemidefiniteRelaxation} \tilde{p}^{\star} = \max_{\vec{y} \in
  \mathbbm{K}^{N_{2 d}}} \sum_{\deg ([w]) \leq 2 d} p_{[w]} y_{[w]},
\end{equation}
such that the constraints~(\ref{Eq:SDPConstraints1})
and~(\ref{Eq:SDPConstraints2}) are satisfied, and $\tilde{p}^{\star}$ is an
upper bound on $p^{\star}$.

\subsection{Symmetric moment relaxations}

We recall that the symmetry group $G^{\pm}_{\star}$ preserves feasibility and
optimality of solutions $(\tmmathbf{X}, \phi) \text{}$. We now consider the
impact of such symmetries on moment relaxations. For that, we note that
$\mathcal{S}_{\sim}$ acts on $\vec{y} = (y_{[w]})$ by
\[ s (\vec{y})_{s ([w])} = y_{[w]}, \]
and, for signed monomials, we define formally $y_{- [w]} = - y_{[w]}$ for $[w]
\in \tilde{\mathcal{W}}$.

\begin{proposition}
  Let $G^{\pm}_{\star}$ be the symmetry group preserving feasibility. Then we
  can add the following constraint to the semidefinite
  relaxation~(\ref{Eq:SemidefiniteRelaxation}):
  \[ \vec{y} = g (\vec{y}), \qquad \text{for all } g \in G^{\pm}_{\star} . \]
\end{proposition}

\begin{proof}
  Let $(\tmmathbf{X}, \phi) \text{}$ be an optimal, feasible, solution, and
  let $(g (\tmmathbf{X}), \phi)$ be an orbit of optimal, feasible, solutions
  under $G^{\pm}_{\star}$. Let $\vec{y} (\tmmathbf{X}, \phi)$ be the solution
  of~(\ref{Eq:SemidefiniteRelaxation}) corresponding to $(\tmmathbf{X},
  \phi)$; then $g (\vec{y}) = \vec{y} (g (\tmmathbf{X}), \phi)$ is also a
  feasible solution of~(\ref{Eq:SemidefiniteRelaxation}) with
  $\tilde{p}^{\star} (\vec{y}) = \tilde{p}^{\star} (g (\vec{y}))$. Now, the
  constraints of the semidefinite program~(\ref{Eq:SemidefiniteRelaxation})
  are convex in $\vec{y}$. So, we can replace any optimal $\vec{y}^{\star}$ by
  \[ \mathcal{R}_{G^\pm_{\star}} (\vec{y}^{\star}) = \frac{1}{| G^{\pm}_{\star} |}
     \sum_{g \in G^\pm_{\star}} g (\vec{y}^{\star}), \]
  to obtain a symmetric optimal solution. By definition,
  $\mathcal{R}_{G^{\pm}_{\star}} (\vec{y}^{\star})$ is invariant under
  $G^{\pm}_{\star}$.
\end{proof}

We now restrict $\vec{y}$ to the symmetric subspace $\vec{y}
=\mathcal{R}_{G^{\pm}_{\star}} (\vec{y})$; thus, we have that
\begin{equation}
  \label{Eq:CanonicalUnderSymmetry} y_{[w]} = y_{\mathcal{C}_{G^{\pm}_{\star}}
  [w]},
\end{equation}
where $\mathcal{C}_{G^{\pm}_{\star}} [w]$ is a canonical representative of
$[w]$ under the symmetry group $G^{\pm}_{\star}$, selected by minimality under graded
lexicographic ordering. Symmetries on the moments $\vec{y}$ translate to
symmetries at the level of the semidefinite matrices. We focus on the main
moment matrix $\Xi$, as the localizing matrices are usually much smaller and
do not appear so frequently in practice\footnote{However, the same procedure
can be applied on $\Lambda_i$; for $\Lambda_i$ corresponding to the constraint
$q_i \succeq 0$, one will need to consider the subgroup $G^\pm_{q_i} \subseteq
G^\pm_{\star}$ that leave this particular $q_i$ invariant as well.}. Remember that
$\Xi$ has rows and columns indexed by the elements $\tilde{\mathcal{W}}$ with
$\Xi_{[r], [c]} = y_{[r^{\ast} c]}$. Then
\[ \Xi_{g ([r]), g ([c])} = y_{[g (r^{\ast}) g (c)]} = y_{[g (r^{\ast} c)]} =
   (g^{- 1} (\vec{y}))_{[r^{\ast} c]} = \vec{y}_{[r^{\ast} c]} = \Xi_{[r],
   [c]} \]
and the matrix $\Xi$ is invariant under the simultaneous permutation of rows
and columns by $G^{\pm}_{\star}$. To cater for generalized permutations in the
above, we define
\[ \Xi_{r, c} = - \Xi_{- r, c} = - \Xi_{r, - c} = \Xi_{- r, - c} \]
for $r, c = 1, \ldots, N_d$. As the matrix $\Xi$ is invariant under the
simultaneous action of a (signed) permutation group on its rows and columns,
it can be block-diagonalized: we refer the reader to~{\cite{Gatermann2004}}
for a clear explanation of the exploitation of such block structures in
semidefinite programs.

\subsection{Example}

We come back to the example of Section~\ref{Sec:CHSHInequality} and consider a
relaxation of degree $1$. We index our moment matrix $\Xi$ using the sequence
of monomials $([1], [A_0], [A_1], [B_0], [B_1])$ such that
\[ \Xi = \left(\begin{array}{ccccc}
     y_{[1]} & y_{[A_0]} & y_{[A_1]} & y_{[B_0]} & y_{[B_1]}\\
     & y_{[1]} & y_{[A_0 A_1]} & y_{[A_0 B_0]} & y_{[A_0 B_1]}\\
     &  & y_{[1]} & y_{[A_1 B_0]} & y_{[A_1 B_1]}\\
     &  &  & y_{[1]} & y_{[B_0 B_1]}\\
     &  &  &  & y_{[1]}
   \end{array}\right), \]
as the matrix is symmetric the elements of the lower triangle are conjugates
of those in the upper triangle. Under symmetrization, we get that
\[ y_{[A_0]} = - y_{[A_0]} = y_{[A_1]} = y_{[B_0]} = y_{[B_1]}, \qquad y_{[A_0
   B_0]} = y_{[A_0 B_1]} = y_{[A_0 B_1]} = - y_{[A_1 B_1]}, \]
\[ y_{[A_0 A_1]} = - y_{[A_0 A_1]}, \qquad y_{[B_0 B_1]} = - y_{[B_0 B_1]}, \]
and thus our semidefinite program simplifies to a program involving a single
variable:
\[ \tilde{p}^{\star} = \max_{x \in \mathbbm{R}} 4 x \]
\[ \text{subject to } \Xi = \left(\begin{array}{ccccc}
     1 & 0 & 0 & 0 & 0\\
     & 1 & 0 & x & x\\
     &  & 1 & x & - x\\
     &  &  & 1 & 0\\
     &  &  &  & 1
   \end{array}\right) \succeq 0. \]
Finally, note that $\Xi$ can be fully diagonalized as $U^{\ast}~ \Xi~ U = \tmop{diag} \left( 1, 1 - \sqrt{2} x, 1 - \sqrt{2} x, 1 +
   \sqrt{2} x, 1 + \sqrt{2} x \right)$ using
\[  \qquad U = \left(\begin{array}{rrrrr}
     1 & 0 & 0 & 0 & 0\\
     0 & - 1 / 2 & - 1 / 2 & 1 / 2 & 1 / 2\\
     0 & 1 / 2 & - 1 / 2 & - 1 / 2 & 1 / 2\\
     0 & 0 & \sqrt{2} / 2 & 0 & \sqrt{2} / 2\\
     0 & \sqrt{2} / 2 & 0 & \sqrt{2} / 2 & 0
   \end{array}\right), \]
and we easily recover the bound $\tilde{p}^{\star} = 2 \sqrt{2}$ on the value
of the CHSH inequality.

\subsection{Generalizations of the NPA hierarchy}
\label{Sec:NPAGeneralization}
The original NPA hierarchy~{\cite{Navascues2007,Navascues2008a}} employs
states $\phi \in \mathcal{H}$ and the evaluation $\langle p (\tmmathbf{X})
\rangle_{\phi} = \langle \phi | p (\tmmathbf{X}) | \phi \rangle$ for the
polynomial $p$. The PPT hierarchy introduced in~{\cite{Moroder2013}} instead
employs density matrices $\rho \in \mathcal{B} \left( \mathcal{H}_{\text{A}}
\otimes \mathcal{H}_{\text{B}} \right)$ with positive partial transpose
($\rho^{\top_{\text{B}}} \succeq 0$), and the evaluation rule
\[ \langle p (\tmmathbf{X}) \rangle_{\rho} = \tmop{tr} [\rho ~ p (\tmmathbf{X})]
   . \]
In our relaxations, it translates to the constraint that
\begin{equation}
  \label{Eq:PPTEvaluation} \mathcal{L} (\alpha \beta) =\mathcal{L} (\alpha
  \beta^{\ast})
\end{equation}
if $\alpha$ (resp. $\beta$) is a product of operators acting only on
$\mathcal{H}_{\text{A}}$ (resp. $\mathcal{H}_{\text{B}}$). The tracial moment
hierarchy~{\cite{Cafuta2011,Navascues2012,Burgdorf2012,Burgdorf2013}} does not
only employ states for the evaluation, but rather defines
\[ \langle p (\tmmathbf{X}) \rangle = \frac{1}{d} \tmop{tr} [p
   (\tmmathbf{X})], \]
which translates to
\begin{equation}
  \label{Eq:TraceEvaluation} \mathcal{L} (f g) =\mathcal{L} (g f)
\end{equation}
for arbitrary $f, g \in \tilde{\mathcal{W}}$ due to the cyclic property of the
trace. A framework for such generalizations is discussed in
Section~\ref{Sec:LinearEvaluation}.

\section{Implementation details}

\label{Sec:Implementation}We built our software library in Scala, a language
that provides four main advantages: it runs on the Java virtual machine (an
optimal combination of portability and speed), it has a strong type system
able to encode mathematical abstractions~{\cite{Oliveira2010}}, it provides a
flexible syntax well-suited to the creation of domain specific
languages~{\cite{DeVito2011}}, and it interfaces with good libraries to
represent exact number types (rational, cyclotomic or algebraic numbers). We
now walk through key parts of our implementation. As the library is under
active development, we refer the user to the up-to-date documentation present
on the repository
\href{https://github.com/denisrosset/symdpoly}{https://github.com/denisrosset/symdpoly}.

\subsection{Definition of the free algebra}

We exploit the syntax of the Scala programming language. We start by defining
the $\ast$-monoid $\mathcal{W}$, by creating an object extending
\tmverbatim{free.MonoidDef}. As
seen by the user, the variables $\tmmathbf{x}= \{ x_1, x_2, \ldots, x_n,
x_1^{\ast}, x^{\ast}_2, \ldots, x_n^{\ast} \}$ are represented by data classes
with an arbitrary number of indices. Internally, however, we work with
integers indexing all possible instances of those variables; the range of
possible instances is provided by the companion object property \tmverbatim{allInstances},
and all operator variables are enumerated in a variable \tmverbatim{operators}.
The \tmverbatim{adjoint} method of each variable returns $x_i^{\ast}$ given $x_i$.
Convenience implementations are provided by the \tmverbatim{HermitianOp} and \tmverbatim{HermitianType\#} base classes.
We use Scala pattern matching to provide readable notation. For the example of Section~\ref{Sec:CHSHInequality}:

\begin{tmcode}
object Free extends free.MonoidDef \{

  case class A(x: Int) extends HermitianOp
  object A extends HermitianType1(0 to 1)

  case class B(y: Int) extends HermitianOp
  object B extends HermitianType1(0 to 1)

  val operators = Seq(A, B)
\}
\end{tmcode}

Internally, we build a table of the adjoints for all possible instances, so future
processing only requires a single array lookup. Monomials are represented by a
length $n$ and an array of integer indices. By convention $n = - 1$ represents
the monomial $0$. Note that the \tmverbatim{Op} type is an inner type of the
object \tmverbatim{Free} written \tmverbatim{Free.Op} (a
path-dependent type), and the Scala type system will make sure that
\tmverbatim{Free.Op} instances are not mixed with variables from other
rings.

\subsection{Definition of the quotient algebra}

The quotient algebra is given by the rewrite rules $R$, expressed naturally
as:
\begin{tmcode}
val Quotient = quotient.MonoidDef(Free) \{
  case (A(x1), A(x2)) if x1 == x2 => Mono.one
  case (B(y1), B(y2)) if y1 == y2 => Mono.one
  case (B(y), A(x)) => A(x) * B(y)
  case (op1, op2) => op1 * op2
\}
\end{tmcode}

For now, we only allow rewrite rules where the left monomial is of degree two,
as they apply to a majority of optimization problems in quantum information.
Internally, we build an integer table \tmverbatim{rule(i,j)}, where
\tmverbatim{i}, \tmverbatim{j} run over indices of the first and second
variable, and the value in \tmverbatim{rule(i,j)} is interpreted as follows:

\begin{tabular}{|p{1.5cm}|p{1.5cm}|p{2.0cm}|p{1.0cm}|p{1.8cm}|p{2.0cm}|}
  \hline
  0 & 1 & 2 & 3 & 4 & 5\\
  \hline
  Set to 0 & Preserve & Remove both & Swap & Keep 1st, remove 2nd &
  Custom\\
  \hline
\end{tabular}

A value $5$/Custom requires an additional lookup in a dictionary, but all
other substitutions are fast. Monomial substitution is performed by Algorithm~\ref{Alg:NormalForm},
so that equivalence classes $[w]$ are represented by their normal form $\mathcal{N}_R
(w) \in \mathcal{W}^0$.

\begin{algorithm}
  \caption{Computation of monomial normal form}
  \label{Alg:NormalForm}
  
  \begin{tmindent}
    Input: length $n$, array of integers $(m_1, \ldots, m_n)$ representing the
    monomial
    
    Output: new length $n$ (or special flag $n = - 1$ indicating zero
    monomial), new array $(m_1, \ldots, m_n)$
    
    \
    
    $i \leftarrow 1$
    
    While $i \leqslant n - 1$
    
    \quad If \tmverbatim{rule}$(m_i, m_{i + 1}) = 0$
    
    \qquad$n \leftarrow - 1$
    
    \qquad Return
    
    \quad ElseIf \tmverbatim{rule}$(m_i, m_{i + 1}) = 1$
    
    \qquad$i \leftarrow i + 1$
    
    \quad Else\quad
    
    \qquad Perform substitution at the $i$-th position, update length $n$
    
    \qquad If $i \neq 4$ Then $i \leftarrow \max (i - 1, 0)$
    
    \quad End
    
    End
  \end{tmindent}
\end{algorithm}

\subsection{Definition of symmetries}

\label{Sec:DefinitionOfSymmetries}Generalized permutation of the variables are
declared again using pattern matching. In our
example~(\ref{Eq:CHSHGenerators}):

\tmtexttt{val p1 = Free.generator \{\\
\quad case A(i) => B(i)\\
\quad case B(i) => A(i)\\
\}\\
val p2 = Free.generator \{\\
\quad case B(0) => B(1)\\
\quad case B(1) => B(0)\\
\quad case A(i) => A(i)\\
\}\\
val p3 = Free.generator \{\\
\quad case A(1) => -A(1)\\
\quad case op => op // fallback, do nothing\\
\}}

The group $\mathcal{S}^\pm_{\sim}$ has to be explicitly constructed by the user:
\begin{tmcode}
val ambientGroup = Quotient.ambientGroup(p1, p2, p3)
\end{tmcode}
and is internally represented as a permutation group on (signed) indices of
variables using a stabilizer chain, see~{\cite{Holt2005}}.

\subsection{Definition of the optimization problem}

Polynomials entering in the optimization problem are defined using standard
mathematical notation. For example:
\begin{alltt}
def A(x: Int) = Quotient.quotient(Free.A(x))
def B(y: Int) = Quotient.quotient(Free.B(y))
val CHSH = A(0)*B(0) + A(0)*B(1) + A(1)*B(0) - A(1)*B(1)
\end{alltt}
Internally, polynomials are represented by the pairs $([w], p_{[w]})$ for
which $p_{[w]} \neq 0$, sorted using graded reverse lexicographic order. For
now, our library does not support constraints of the form $q_i (\tmmathbf{X})
\succeq 0$, $r_j (\tmmathbf{X}) = 0$ or $\langle s_k (\tmmathbf{X})
\rangle_{\phi} \geq 0$; thus only an objective polynomial $p$ can be provided
for now. This implies that $G^\pm = S^\pm_{\sim}$.

\subsection{Linear evaluation: rules and canonical form}

\label{Sec:LinearEvaluation}
The constraints~(\ref{Eq:LinearFunctional})
applying on the linear functional $\mathcal{L}$, with the possible addition
of~(\ref{Eq:PPTEvaluation}) or (\ref{Eq:TraceEvaluation}), do not apply at the
level of monomials, but only when performing the final evaluation.

Ncause our polynomials have real coefficients, we can have force $\vec{y}$ to be real as well and thus
$\vec{y}_{[w]} = \vec{y}_{[w^{\ast}]}$. That corresponds to invariance under global transposition.
Additional equivalence relations can be specified in the code, and correspond of predicate of two types:
\begin{itemizeminus}
  \item Transposition equivalence under a predicate $\mathcal{P}: \tmmathbf{x}
  \rightarrow \{ \tmop{true}, \tmop{false} \}$, that apply in place to the
  variables for which the predicate is true.
  
  \item Cyclic permutations under a predicate $\mathcal{P}: \tmmathbf{x}
  \rightarrow \{ \tmop{true}, \tmop{false} \}$, that apply in place to the
  variables for which the predicate is true.
\end{itemizeminus}
Let $w = t_1 t_2 f_1 t_3 f_2 f_3 t_4$ such that the predicate $\mathcal{P}$ is
true for the variables $t_i$ and false for the variables $f_i$. Then,
application of a transposition returns
\[ T_{\mathcal{P}} (w) = t_4^{\ast} t_3^{\ast} f_1 t_2^{\ast} f_2 f_3
   t_1^{\ast}, \]
while a single application of a cyclic permutation returns
\[ C_{\mathcal{P}} (w) = t_2 t_3 f_1 t_4 f_2 f_3 t_1 . \]
The canonical form $\mathcal{C} (w) \in \mathcal{W}^{\pm}$ of a signed
monomial $w \in \mathcal{W}^{\pm}$ is obtained by applying the rewriting rules
$R$ on all iterations of
\begin{itemizeminus}
  \item the symmetry group $G^\pm_{\star}$,
  
  \item partial transpositions $w \rightarrow \{ w, T_{\mathcal{P}} (w) \}$
  for all transposition predicates,
  
  \item cyclic permutations $w \rightarrow \{ w, C_{\mathcal{P}} (w),
  C_{\mathcal{P}} (C_{\mathcal{P}} (w)), \ldots \}$ for all cyclic permutation
  predicates,
\end{itemizeminus}
and keeping the minimal lexicographic representative along with its sign. In
the case that $\mathcal{C} (w) =\mathcal{C} (- w)$, we set formally
$\mathcal{C} (w) = 0$. For the problem sizes considered, brute force
evaluation is faster than algorithms exploiting the problem structure,
provided the code is optimized to operate on primitive types (machine-size
integers) during the enumeration. In our library, predicates are defined using
pattern matching.

\tmtexttt{val partialTransposeBob = PartialTransposition(Free) \{\\
\quad case A(i) => false\\
\quad case B(i) => true\\
\}\\
val fullCyclic = CyclicPermutation(Free) \{\\
\quad case op => true\\
\}}

\subsection{Computation of the symmetry group $G^\pm_{\star}$}

The group $G^\pm$ is provided by the user, see
Section~\ref{Sec:DefinitionOfSymmetries}. To construct the symmetry subgroup
$G^\pm_{\star}$ that preserves the objective value $p \in \mathbbm{K}
[\tilde{\mathcal{W}}]$, we proceed as follow.

We construct the smallest set of signed monomials $M \subset
\tilde{\mathcal{W}}^{\pm}$ that
\begin{itemizeminus}
  \item includes all monomials present in $p$,
  
  \item is invariant under action of $G^\pm$ (i.e. $w \in M \Leftrightarrow g (w)
  \in M$).
\end{itemizeminus}
We write $\mathcal{S}_M$ the symmetric group acting on $M$. As $G^\pm$ acts on $M$
as well, there exists a permutation group $H \subseteq S_M$, along with an
isomorphism $\varphi : G^\pm \rightarrow H$ representing this action.

Explicit steps are presented in Algorithm~\ref{Alg:SymmetryGroup}.
Fast algorithms based on stabilizer chains exist for the last two steps of the
algorithm~{\cite{Holt2005}}, and are implemented in GAP
System~{\cite{Group2007}} or our library Alasc~{\cite{Rosset2017d}}.

\begin{algorithm}
  \caption{Computation of the symmetry group $G^\pm_{\star}$}
  \label{Alg:SymmetryGroup}

Input: Objective polynomial $p$, set of monomials $M \subset
\tilde{\mathcal{W}}$, group $H$, isomorphism $\varphi$

Output: symmetry subgroup $G^\pm_{\star}$

\

Compute the partition $\mathcal{P}$ of $\mathcal{M}$ under the equivalence
relation $v^{\pm} \sim w^{\pm}$ if $y_{[\mathcal{C} (v^{\pm})]} =
y_{[\mathcal{C} (w^{\pm})]}$.

Compute the subgroup of $H_{\star} \subseteq H$ that stabilizes $\mathcal{P}$.

Return $G^\pm_{\star} = \varphi^{- 1} (H_{\star})$.
\end{algorithm}

\subsection{Construction of the symmetrized moment matrices}

We come to the final part of our method, the construction of moment matrices
to be provided to the solver. We use a value $T$ as a token, where $T$ is
larger than a crude upper bound on the number of final monomials, for example,
$T = 2^{31} - 1$.
The full method is presented in Algorithm~\ref{Alg:ReduceMoments}.
At the output, the matrices $C$ and $\{ A_k \}$ of the SDP are recovered with:
\[ C_{j k} = \left\{ \begin{array}{ll}
     1 & \text{if } J_{j k} = 1\\
     0 & \text{otherwise}
   \end{array} \right., \qquad (A_i)_{j k} = \left\{ \begin{array}{ll}
     \tmop{sign} (J_{j k}) & \text{if } | J_{j k} | = i\\
     0 & \text{otherwise}
   \end{array} \right. . \]
Note that with our convention, the variable $y_1$ is never used and $A_1 = 0$.
The variables $y_2, y_3, \ldots$ correspond to the unique moments that appear
in the SDP matrix according to the map $\mu$. To get the expression of the
objective in the symmetrized variables, initialize $\vec{b} \leftarrow \vec{0}
\in \mathbbm{R}^{N_M}$, and for each term $p_{[w]} [w]$ appearing in the
objective with real coefficient $p_{[w]}$ and monomial $[w]$, set
$b_{\mathcal{C} (w)} \leftarrow b_{\mathcal{C} (w)} + p_{[w]}$.

The resulting SDP program is thus:
\[ \begin{array}{rl}
     \max & \vec{b}^{\top} \cdot \vec{y}\\
     \text{over} & \vec{y} \in \mathbbm{R}^{N_M}\\
     \text{such that} & \chi = C - \sum_i y_i A_i \geqslant 0
   \end{array} \]
where $C$ and $\{ A_i \}$ are symmetric matrices in $\mathbbm{R}^{n \times
n}$.

\begin{algorithm}
  \caption{Computation of the symmetrized moment matrices}
  \label{Alg:ReduceMoments}
{\tmstrong{Input}}

List of generating monomials $V = (v_i)$, $i = 1, \ldots, N_D$

Symmetry group $G^\pm_{\star}$

Predicates of partial transpositions $\{ \mathcal{P}^T_i \}$, cyclic
permutations $\{ \mathcal{P}^C_i \}$

\

{\tmstrong{Output}}

$J$ a matrix of integers of size $n \times n$

$\mu$ a bidirectional map between operator sequences and integer indices

$N_M$ total number of matrices $\{ C, A_i \}$ in the SDP constraint

\

Initialize $J$ with the token value $T$

Initialize $\mu$ empty

$N_M \leftarrow 1$

For $i = 1, \ldots, n, j = i, \ldots, n$

\quad If $J_{i, j} \neq T$

\qquad Skip the current iteration, the current cell has already been computed

\quad End

\quad Compute the canonical $c =\mathcal{C} (N_R (v_i^{\dag} v_j))$ by
enumeration (global transposition, $G^\pm_{\star}$, $\{ \mathcal{P}^T_i \}$, $\{
\mathcal{P}^C_i \}$)

\quad If $c \in \{ 0, 1 \}$

\qquad$k \leftarrow c$

\quad ElseIf $\mu (c)$ is defined

\qquad$k \leftarrow \mu (c)$

\quad Else

\qquad$N_M \leftarrow N_M + 1$

\qquad$\mu (c) \leftarrow N_M$

\quad End

\quad For $(r, c) \in \{ (g (i), g (j)) : g \in G^\pm_{\star} \}$

\qquad$\sigma \leftarrow \tmop{sign} (r c)$

\qquad$J_{| r |, | c |} \leftarrow \sigma k$

\qquad$J_{| c |, | r |} \leftarrow \sigma k$

\quad End

End
\end{algorithm}

\section{Application: high precision bounds for the $I_{3322}$ inequality}

\label{Sec:Application}As a test of our technique, we apply our symmetrization
technique to the $I_{3322}$
inequality~{\cite{Froissart1981,Sliwa2003,Collins2004}}, in its variant
symmetric under permutation of parties:
\[ I_{3322} = (A_1 B_3 + A_3 B_1 - A_2 B_3 - A_3 B_2) + (A_1 B_1 - A_1 B_2 -
   A_2 B_1 + A_2 B_2) + A_1 + A_2 + B_1 + B_2, \]
which, in particular, is symmetric under permutation of parties, and under the
signed permutation:
\[ (A_1, A_2, A_3, B_1, B_2, B_3) \rightarrow (A_2, A_1, A_3, B_1, B_2, - B_3)
\]
which together generate a symmetry group of order $8$, which we identify as a
dihedral group. We compare three approaches:
\begin{itemizeminus}
  \item Symmetries=''no'': we construct the standard NPA relaxation.
  
  \item Symmetries=''partial'': we symmetrize our SDP as discussed in
  Section~\ref{Sec:Implementation}, and then split the blocks according to the
  symmetric/antisymmetric subspace under permutation of parties, which is the
  most simple block diagonalization possible.
  
  \item Symmetries=''diag'': we symmetrize our SDP, and then try to split the
  blocks as much as possible. Unfortunately, there is no general software
  package to perform full block-diagonalization in exact arithmetic. Thus, we
  performed the block diagonalization by hand. This explains the presence of
  $6$ blocks in our results, whereas the dihedral group of order $8$ has five
  rational representations, and thus we should expect a decomposition in at
  most 5 blocks.
\end{itemizeminus}
We see that the biggest reductions in memory usage are provided by the reduction of
the number of variables $m$ and the quite straightforward
block-diagonalization due to symmetry under permutation of parties; this is
not surprising as the memory usage is in general dominated by a factor
$\mathcal{O} (m^2)$. However, the block diagonalization has still a
non negligible impact on the CPU time, as it reduces the terms in $\mathcal{O}
(n^3)$. Computational results using the SDPA double-double precision
solver~{\cite{Nakata2010}} are reported in Table~\ref{Table:Comparison}.

\begin{table}[!ht]
  \begin{tabular}{|l|l|l|l|l|l|}
    \hline
    Level & Sym. & \# vars & SDP block sizes & CPU time (s) &
    Memory use (MB)\\
    \hline
    3 & no & 867 & 88 & 73.2 & 15\\
    \hline
    3 & partial & 124 & 44,44 & 4.3 & 3\\
    \hline
    3 & diag. & 124 & 22,22,13,11,11,9 & 1.2 & 2\\
    \hline
    4 & no & 4491 & 244 & 8416.0 & 331\\
    \hline
    4 & partial & 593 & 122,122 & 292.9 & 14\\
    \hline
    4 & diag. & 593 & 61,61,35,31,30,26 & 62.6 & 10\\
    \hline
  \end{tabular}
  \caption{\label{Table:Comparison}Relaxation levels, symmetry reduction
  method used and resources needed to solve successfully the SDP relaxation using the
  SDPA double-double precision solver.}
\end{table}

For completeness, we also computed the relaxation of level 5 using a
diagonalization in 4 blocks (sizes: 162+157+157+152), which completed in 8700
[s] using SDPA in double-double precision. We also recomputed levels 2-4 in
quadruple-double precision. The results are, along with the gap being the
difference between the objective value of the primal and dual problem reported
by the solver:
\begin{equation}
  I_2 \cong 1.2509397216370581 \text{ (gap} \sim 10^{- 31} \text{)},
\end{equation}
\begin{equation}
  I_3 \cong 1.2508755620230350 \text{ (gap} \sim 10^{- 31} \text{)},
\end{equation}
\begin{equation}
  I_4 \cong 1.2508753845139768 \text{ (gap} \sim 10^{- 30} \text{)},
\end{equation}
\begin{equation}
  I_5 \cong 1.2508753845139766 \text{ (gap} \sim 10^{- 21} \text{)} .
\end{equation}
There still seems to be a gap between $I_4$ and $I_5$, but the SDPA high
precision solvers only report $\sim 17$ digits. Using the VSDP
package~{\cite{Jansson2006}}, rigorous bounds can be computed for the
solution of a semidefinite program. Because the symmetrization reduces the
complexity of the problem, we should be able to observe an effect on the
numerical precision of the obtained results. Indeed, we computed a robust
solution for $I_3$ using VSDP and standard double precision arithmetic, to
obtain:
\begin{equation}
  I_3^{\text{without-symmetrization}} \in \left[ \tmstrong{1.2508755} 5,
  \tmstrong{1.2508755} 7 \right],
\end{equation}
and
\begin{equation}
  I_3^{\text{with-symmetrization}} \in \left[ \tmstrong{1.25087556} 1,
  \tmstrong{1.25087556} 8 \right],
\end{equation}
and we see that symmetrization provides an additional digit of precision.

\section{Conclusion}

We introduced symmetry-adapted moment relaxations for optimization problems
over noncommutative polynomials, with a particular emphasis on the problems
arising from quantum information scenarios. We also presented a software
library automating the discovery and use of the symmetries present in the
problem formulation. This work is only a first step in that journey. In
particular, we are looking to extend our software library in the following
directions. First, the code has not been tested on problems involving
non-Hermitian variables and polynomial with complex coefficients. Second, we
are lacking implementations of localizing matrices and support for general
linear constraints. These additions should be pretty straightforward, except
that automatic discovery of the full symmetry group could be difficult under
involved combinations of constraints. Additional gains can be
obtained using block diagonalization. As of today, there exists a variety of
algorithms to decompose an algebra of matrices commuting with the
representation of a group: the software library AREP~{\cite{Puschel2002}}
provides structured decomposition of permutation representations of solvable
groups; numerical approaches can decompose arbitrary
representations~{\cite{Maehara2010}}; finally, a recent preprint proposed an
algebraic method~{\cite{Kornyak2018}} based on Groebner bases. We also mention
a recent promising approach using Jordan algebras~{\cite{Permenter2016}}.

{\noindent}\tmtextbf{Acknowledgements . }We thank many colleagues for
discussions in the last four years; among them Jean-Daniel Bancal, Nicolas
Gisin, Yeong-Cherng Liang, Alejandro Pozas, Marc-Olivier Renou, Armin
Tavakoli, Jamie Sikora, Peter Wittek and Elie Wolfe. This work was supported
by the Swiss National Science Foundation via the Mobility Fellowship
P2GEP2\_162060, by the Perimeter Institute for Theoretical Physics. Research
at Perimeter Institute is supported by the Government of Canada through the
Department of Innovation, Science and Economic Development and by the Province
of Ontario through the Ministry of Research and Innovation. This publication
was made possible through the support of a grant from the John Templeton
Foundation. We also acknowledge support by the Ministry of Education, Taiwan,
R.O.C., through Aiming for the Top University Project granted to the National
Cheng Kung University (NCKU), and by the Ministry of Science and Technology,
Taiwan (Grants No. 104-2112-M-006-021-MY3).{\hspace*{\fill}}{\medskip}
\bibliography{symdpoly}

\begin{thebibliography}{10}

\bibitem{Bamps2015}
C\'edric Bamps and Stefano Pironio.
\newblock Sum-of-squares decompositions for a family of
  {{Clauser}}-{{Horne}}-{{Shimony}}-{{Holt}}-like inequalities and their
  application to self-testing.
\newblock {\em Phys. Rev. A}, 91(5):052111, May 2015.
\newblock \href {http://dx.doi.org/10.1103/PhysRevA.91.052111}
  {\path{doi:10.1103/PhysRevA.91.052111}}.

\bibitem{Bancal2015}
Jean-Daniel Bancal, Miguel Navascu\'es, Valerio Scarani, Tam\'as V\'ertesi, and
  Tzyh~Haur Yang.
\newblock Physical characterization of quantum devices from nonlocal
  correlations.
\newblock {\em Phys. Rev. A}, 91(2):022115, February 2015.
\newblock \href {http://dx.doi.org/10.1103/PhysRevA.91.022115}
  {\path{doi:10.1103/PhysRevA.91.022115}}.

\bibitem{Branczyk2007}
Agata~M. Bra\'nczyk, Paulo E. M.~F. Mendon{\c c}a, Alexei Gilchrist, Andrew~C.
  Doherty, and Stephen~D. Bartlett.
\newblock Quantum control of a single qubit.
\newblock {\em Phys. Rev. A}, 75(1):012329, January 2007.
\newblock \href {http://dx.doi.org/10.1103/PhysRevA.75.012329}
  {\path{doi:10.1103/PhysRevA.75.012329}}.

\bibitem{Burgdorf2013}
Sabine Burgdorf, Kristijan Cafuta, Igor Klep, and Janez Povh.
\newblock Algorithmic aspects of sums of {{Hermitian}} squares of
  noncommutative polynomials.
\newblock {\em Comput Optim Appl}, 55(1):137--153, May 2013.
\newblock \href {http://dx.doi.org/10.1007/s10589-012-9513-8}
  {\path{doi:10.1007/s10589-012-9513-8}}.

\bibitem{Burgdorf2012}
Sabine Burgdorf and Igor Klep.
\newblock The truncated tracial moment problem.
\newblock {\em J. Oper. Theory}, pages 141--163, 2012.

\bibitem{Cafuta2011}
Kristijan Cafuta, Igor Klep, and Janez Povh.
\newblock {{NCSOStools}}: A computer algebra system for symbolic and numerical
  computation with noncommutative polynomials.
\newblock {\em Optimization Methods and Software}, 26(3):363--380, June 2011.
\newblock \href {http://dx.doi.org/10.1080/10556788.2010.544312}
  {\path{doi:10.1080/10556788.2010.544312}}.

\bibitem{Cavalcanti2017b}
D.~Cavalcanti and P.~Skrzypczyk.
\newblock Quantum steering: A review with focus on semidefinite programming.
\newblock {\em Rep. Prog. Phys.}, 80(2):024001, 2017.
\newblock \href {http://dx.doi.org/10.1088/1361-6633/80/2/024001}
  {\path{doi:10.1088/1361-6633/80/2/024001}}.

\bibitem{Chiribella2012}
Giulio Chiribella.
\newblock Optimal networks for quantum metrology: Semidefinite programs and
  product rules.
\newblock {\em New J. Phys.}, 14(12):125008, 2012.
\newblock \href {http://dx.doi.org/10.1088/1367-2630/14/12/125008}
  {\path{doi:10.1088/1367-2630/14/12/125008}}.

\bibitem{Clauser1969}
John~F. Clauser, Michael~A. Horne, Abner Shimony, and Richard~A. Holt.
\newblock Proposed {{Experiment}} to {{Test Local Hidden}}-{{Variable
  Theories}}.
\newblock {\em Phys. Rev. Lett.}, 23(15):880--884, October 1969.
\newblock \href {http://dx.doi.org/10.1103/PhysRevLett.23.880}
  {\path{doi:10.1103/PhysRevLett.23.880}}.

\bibitem{Collins2004}
Daniel Collins and Nicolas Gisin.
\newblock A relevant two qubit {{Bell}} inequality inequivalent to the {{CHSH}}
  inequality.
\newblock {\em J. Phys. A: Math. Gen.}, 37(5):1775, February 2004.
\newblock \href {http://dx.doi.org/10.1088/0305-4470/37/5/021}
  {\path{doi:10.1088/0305-4470/37/5/021}}.

\bibitem{DeVito2011}
Zachary DeVito, Niels Joubert, Francisco Palacios, Stephen Oakley, Montserrat
  Medina, Mike Barrientos, Erich Elsen, Frank Ham, Alex Aiken, Karthik
  Duraisamy, Eric Darve, Juan Alonso, and Pat Hanrahan.
\newblock Liszt: {{A Domain Specific Language}} for {{Building Portable
  Mesh}}-based {{PDE Solvers}}.
\newblock In {\em Proceedings of 2011 {{International Conference}} for {{High
  Performance Computing}}, {{Networking}}, {{Storage}} and {{Analysis}}}, SC
  '11, pages 9:1--9:12, New York, NY, USA, 2011. {ACM}.
\newblock \href {http://dx.doi.org/10.1145/2063384.2063396}
  {\path{doi:10.1145/2063384.2063396}}.

\bibitem{Doherty2008}
A.C. Doherty, Yeong-Cherng Liang, B.~Toner, and S.~Wehner.
\newblock The {{Quantum Moment Problem}} and {{Bounds}} on {{Entangled
  Multi}}-prover {{Games}}.
\newblock In {\em 23rd {{Annual IEEE Conference}} on {{Computational
  Complexity}}, 2008. {{CCC}} '08}, pages 199--210, June 2008.
\newblock \href {http://dx.doi.org/10.1109/CCC.2008.26}
  {\path{doi:10.1109/CCC.2008.26}}.

\bibitem{Doherty2004}
Andrew~C. Doherty, Pablo~A. Parrilo, and Federico~M. Spedalieri.
\newblock Complete family of separability criteria.
\newblock {\em Phys. Rev. A}, 69(2):022308, February 2004.
\newblock \href {http://dx.doi.org/10.1103/PhysRevA.69.022308}
  {\path{doi:10.1103/PhysRevA.69.022308}}.

\bibitem{Drazin1978}
Michael~P. Drazin.
\newblock Natural structures on semigroups with involution.
\newblock {\em Bull. Amer. Math. Soc.}, 84(1):139--141, January 1978.

\bibitem{Eldar2003}
Y.~C. Eldar.
\newblock A semidefinite programming approach to optimal unambiguous
  discrimination of quantum states.
\newblock {\em IEEE Trans. Inf. Theory}, 49(2):446--456, February 2003.
\newblock \href {http://dx.doi.org/10.1109/TIT.2002.807291}
  {\path{doi:10.1109/TIT.2002.807291}}.

\bibitem{Fadel2017}
Matteo Fadel and Jordi Tura.
\newblock Bounding the {{Set}} of {{Classical Correlations}} of a
  {{Many}}-{{Body System}}.
\newblock {\em Phys. Rev. Lett.}, 119(23):230402, December 2017.
\newblock \href {http://dx.doi.org/10.1103/PhysRevLett.119.230402}
  {\path{doi:10.1103/PhysRevLett.119.230402}}.

\bibitem{Ferrier1998}
Ch~Ferrier.
\newblock Hilbert's 17th problem and best dual bounds in quadratic
  minimization.
\newblock {\em Cybern Syst Anal}, 34(5):696--709, September 1998.
\newblock \href {http://dx.doi.org/10.1007/BF02667043}
  {\path{doi:10.1007/BF02667043}}.

\bibitem{Froissart1981}
M.~Froissart.
\newblock Constructive generalization of {{Bell}}'s inequalities.
\newblock {\em Nuov Cim B}, 64(2):241--251, August 1981.
\newblock \href {http://dx.doi.org/10.1007/BF02903286}
  {\path{doi:10.1007/BF02903286}}.

\bibitem{Gatermann2004}
Karin Gatermann and Pablo~A. Parrilo.
\newblock Symmetry groups, semidefinite programs, and sums of squares.
\newblock {\em Journal of Pure and Applied Algebra},
  192(1\textendash{}3):95--128, September 2004.
\newblock \href {http://dx.doi.org/10.1016/j.jpaa.2003.12.011}
  {\path{doi:10.1016/j.jpaa.2003.12.011}}.

\bibitem{Group2007}
GAP Group and {others}.
\newblock Gap system for computational discrete algebra.
\newblock 2007.

\bibitem{Helton2002}
J.~William Helton.
\newblock "{{Positive}}" {{Noncommutative Polynomials Are Sums}} of
  {{Squares}}.
\newblock {\em Ann. Math.}, 156(2):675--694, 2002.
\newblock \href {http://dx.doi.org/10.2307/3597203}
  {\path{doi:10.2307/3597203}}.

\bibitem{Henrion2009}
Didier Henrion, Jean-Bernard Lasserre, and Johan L\"ofberg.
\newblock {{GloptiPoly}} 3: Moments, optimization and semidefinite programming.
\newblock {\em Optim. Methods Softw.}, 24(4-5):761--779, October 2009.
\newblock \href {http://dx.doi.org/10.1080/10556780802699201}
  {\path{doi:10.1080/10556780802699201}}.

\bibitem{Holt2005}
Derek~F. Holt, Bettina Eick, and Eamonn~A. O'Brien.
\newblock {\em Handbook of {{Computational Group Theory}}}.
\newblock {CRC Press}, January 2005.

\bibitem{Jansson2006}
Christian Jansson.
\newblock {{VSDP}}: {{A MATLAB}} software package for verified semidefinite
  programming.
\newblock {\em DELTA}, 1:4, 2006.

\bibitem{Kornyak2018}
Vladimir~V. Kornyak.
\newblock An {{Algorithm}} to {{Decompose Permutation Representations}} of
  {{Finite Groups}}: {{Polynomial Algebra Approach}}.
\newblock {\em ArXiv:1801.09786 Cs Math}, January 2018.

\bibitem{Lasserre2001}
Jean~B. Lasserre.
\newblock Global {{Optimization}} with {{Polynomials}} and the {{Problem}} of
  {{Moments}}.
\newblock {\em SIAM J. Optim.}, 11:796--817, 2001.

\bibitem{Lee2015}
Ciar\'an~M. Lee and Robert~W. Spekkens.
\newblock Causal inference via algebraic geometry: Necessary and sufficient
  conditions for the feasibility of discrete causal models.
\newblock {\em ArXiv:1506.03880 Quant-Ph Stat}, June 2015.

\bibitem{Liang2007a}
Yeong-Cherng Liang and Andrew~C. Doherty.
\newblock Bounds on quantum correlations in {{Bell}}-inequality experiments.
\newblock {\em Phys. Rev. A}, 75(4):042103, April 2007.
\newblock \href {http://dx.doi.org/10.1103/PhysRevA.75.042103}
  {\path{doi:10.1103/PhysRevA.75.042103}}.

\bibitem{Lofberg2009a}
J.~Lofberg.
\newblock Pre- and {{Post}}-{{Processing Sum}}-of-{{Squares Programs}} in
  {{Practice}}.
\newblock {\em IEEE Trans. Autom. Control}, 54(5):1007--1011, May 2009.
\newblock \href {http://dx.doi.org/10.1109/TAC.2009.2017144}
  {\path{doi:10.1109/TAC.2009.2017144}}.

\bibitem{Maehara2010}
Takanori Maehara and Kazuo Murota.
\newblock A numerical algorithm for block-diagonal decomposition of matrix
  *-algebras with general irreducible components.
\newblock {\em Japan J. Indust. Appl. Math.}, 27(2):263--293, September 2010.
\newblock \href {http://dx.doi.org/10.1007/s13160-010-0007-8}
  {\path{doi:10.1007/s13160-010-0007-8}}.

\bibitem{Moroder2013}
Tobias Moroder, Jean-Daniel Bancal, Yeong-Cherng Liang, Martin Hofmann, and
  Otfried G\"uhne.
\newblock Device-{{Independent Entanglement Quantification}} and {{Related
  Applications}}.
\newblock {\em Phys. Rev. Lett.}, 111(3):030501, July 2013.
\newblock \href {http://dx.doi.org/10.1103/PhysRevLett.111.030501}
  {\path{doi:10.1103/PhysRevLett.111.030501}}.

\bibitem{Nakata2010}
Maho Nakata.
\newblock A numerical evaluation of highly accurate multiple-precision
  arithmetic version of semidefinite programming solver:
  {{SDPA}}-{{GMP}},-{{QD}} and-{{DD}}.
\newblock In {\em Computer-{{Aided Control System Design}} ({{CACSD}}), 2010
  {{IEEE International Symposium}} On}, pages 29--34, 2010.

\bibitem{Navascues2015a}
Miguel Navascu\'es, Adrien Feix, Mateus Ara\'ujo, and Tam\'as V\'ertesi.
\newblock Characterizing finite-dimensional quantum behavior.
\newblock {\em Phys. Rev. A}, 92(4):042117, October 2015.
\newblock \href {http://dx.doi.org/10.1103/PhysRevA.92.042117}
  {\path{doi:10.1103/PhysRevA.92.042117}}.

\bibitem{Navascues2007}
Miguel Navascu\'es, Stefano Pironio, and Antonio Ac\'in.
\newblock Bounding the {{Set}} of {{Quantum Correlations}}.
\newblock {\em Phys. Rev. Lett.}, 98(1):010401, January 2007.
\newblock \href {http://dx.doi.org/10.1103/PhysRevLett.98.010401}
  {\path{doi:10.1103/PhysRevLett.98.010401}}.

\bibitem{Navascues2008a}
Miguel Navascu\'es, Stefano Pironio, and Antonio Ac\'in.
\newblock A convergent hierarchy of semidefinite programs characterizing the
  set of quantum correlations.
\newblock {\em New J. Phys.}, 10(7):073013, July 2008.
\newblock \href {http://dx.doi.org/10.1088/1367-2630/10/7/073013}
  {\path{doi:10.1088/1367-2630/10/7/073013}}.

\bibitem{Navascues2012}
Miguel Navascu\'es, Stefano Pironio, and Antonio Ac\'in.
\newblock {{SDP Relaxations}} for {{Non}}-{{Commutative Polynomial
  Optimization}}.
\newblock In Miguel~F. Anjos and Jean~B. Lasserre, editors, {\em Handbook on
  {{Semidefinite}}, {{Conic}} and {{Polynomial Optimization}}}, number 166 in
  International Series in Operations Research \& Management Science, pages
  601--634. {Springer US}, 2012.
\newblock \href {http://dx.doi.org/10.1007/978-1-4614-0769-0_21}
  {\path{doi:10.1007/978-1-4614-0769-0_21}}.

\bibitem{Navascues2014}
Miguel Navascu\'es and Tam\'as V\'ertesi.
\newblock Bounding the {{Set}} of {{Finite Dimensional Quantum Correlations}}.
\newblock {\em Phys. Rev. Lett.}, 115(2):020501, July 2015.
\newblock \href {http://dx.doi.org/10.1103/PhysRevLett.115.020501}
  {\path{doi:10.1103/PhysRevLett.115.020501}}.

\bibitem{Nesterov2000}
Yurii Nesterov.
\newblock Squared {{Functional Systems}} and {{Optimization Problems}}.
\newblock In {\em High {{Performance Optimization}}}, Applied Optimization,
  pages 405--440. {Springer, Boston, MA}, 2000.
\newblock \href {http://dx.doi.org/10.1007/978-1-4757-3216-0_17}
  {\path{doi:10.1007/978-1-4757-3216-0_17}}.

\bibitem{Oliveira2010}
Bruno~CdS Oliveira, Adriaan Moors, and Martin Odersky.
\newblock Type classes as objects and implicits.
\newblock In {\em {{ACM Sigplan Notices}}}, volume~45, pages 341--360, 2010.

\bibitem{sostools}
Antonis Papachristodoulou, James Anderson, Giorgio Valmorbida, Stephen Prajna,
  Pete Seiler, and Pablo Parrilo.
\newblock {\em {{SOSTOOLS}}: {{Sum}} of Squares Optimization Toolbox for
  {{MATLAB}}}.
\newblock http://arxiv.org/abs/1310.4716, 2013.
\newblock Available from http://www.eng.ox.ac.uk/control/sostools,
  http://www.cds.caltech.edu/sostools and
  http://www.mit.edu/\~parrilo/sostools.

\bibitem{Parrilo2003}
Pablo~A. Parrilo.
\newblock Semidefinite programming relaxations for semialgebraic problems.
\newblock {\em Math. Program., Ser. B}, 96(2):293--320, May 2003.
\newblock \href {http://dx.doi.org/10.1007/s10107-003-0387-5}
  {\path{doi:10.1007/s10107-003-0387-5}}.

\bibitem{Parrilo2005}
Pablo~A. Parrilo.
\newblock Exploiting {{Algebraic Structure}} in {{Sum}} of {{Squares
  Programs}}.
\newblock In {\em Positive {{Polynomials}} in {{Control}}}, Lecture Notes in
  Control and Information Science, pages 181--194. {Springer, Berlin,
  Heidelberg}, 2005.
\newblock \href {http://dx.doi.org/10.1007/10997703_11}
  {\path{doi:10.1007/10997703_11}}.

\bibitem{Permenter2016}
Frank Permenter and Pablo~A. Parrilo.
\newblock Dimension reduction for semidefinite programs via {{Jordan}}
  algebras.
\newblock {\em ArXiv:1608.02090 Math}, August 2016.

\bibitem{Puschel2002}
Markus P\"uschel.
\newblock Decomposing {{Monomial Representations}} of {{Solvable Groups}}.
\newblock {\em Journal of Symbolic Computation}, 34(6):561--596, December 2002.
\newblock \href {http://dx.doi.org/10.1006/jsco.2002.0566}
  {\path{doi:10.1006/jsco.2002.0566}}.

\bibitem{Rains2001}
E.~M. Rains.
\newblock A semidefinite program for distillable entanglement.
\newblock {\em IEEE Trans. Inf. Theory}, 47(7):2921--2933, November 2001.
\newblock \href {http://dx.doi.org/10.1109/18.959270}
  {\path{doi:10.1109/18.959270}}.

\bibitem{Riener2012}
Cordian Riener, Thorsten Theobald, Lina~Jansson Andr\'en, and Jean~B. Lasserre.
\newblock Exploiting {{Symmetries}} in {{SDP}}-{{Relaxations}} for {{Polynomial
  Optimization}}.
\newblock {\em Mathematics of OR}, 38(1):122--141, October 2012.
\newblock \href {http://dx.doi.org/10.1287/moor.1120.0558}
  {\path{doi:10.1287/moor.1120.0558}}.

\bibitem{Rosset2017d}
Denis Rosset.
\newblock Alasc, computational group theory in {{Scala}}.
\newblock {\em GitHub Repos.}, 2017.

\bibitem{Rosset2017a}
Denis Rosset, Nicolas Gisin, and Elie Wolfe.
\newblock Universal bound on the cardinality of local hidden variables in
  networks.
\newblock {\em ArXiv:1709.00707 Quant-Ph}, September 2017.

\bibitem{Sliwa2003}
Cezary \'Sliwa.
\newblock Symmetries of the {{Bell}} correlation inequalities.
\newblock {\em Physics Letters A}, 317(3-4):165--168, October 2003.
\newblock \href {http://dx.doi.org/10.1016/S0375-9601(03)01115-0}
  {\path{doi:10.1016/S0375-9601(03)01115-0}}.

\bibitem{Tavakoli2018a}
Armin Tavakoli, Denis Rosset, and Marc-Olivier Renou.
\newblock Enabling computation of correlation bounds for finite-dimensional
  quantum systems via symmetrisation.
\newblock {\em ArXiv:1808.02412 Quant-Ph}, August 2018.

\bibitem{Vidal2002}
G.~Vidal and R.~F. Werner.
\newblock Computable measure of entanglement.
\newblock {\em Phys. Rev. A}, 65(3):032314, February 2002.
\newblock \href {http://dx.doi.org/10.1103/PhysRevA.65.032314}
  {\path{doi:10.1103/PhysRevA.65.032314}}.

\bibitem{Waki2008}
Hayato Waki, Sunyoung Kim, Masakazu Kojima, Masakazu Muramatsu, and Hiroshi
  Sugimoto.
\newblock Algorithm 883: {{SparsePOP}}\textemdash{{A Sparse Semidefinite
  Programming Relaxation}} of {{Polynomial Optimization Problems}}.
\newblock {\em ACM Trans Math Softw}, 35(2):15:1--15:13, July 2008.
\newblock \href {http://dx.doi.org/10.1145/1377612.1377619}
  {\path{doi:10.1145/1377612.1377619}}.

\bibitem{Wittek2015}
Peter Wittek.
\newblock Algorithm 950: {{Ncpol2Sdpa}}\textemdash{{Sparse Semidefinite
  Programming Relaxations}} for {{Polynomial Optimization Problems}} of
  {{Noncommuting Variables}}.
\newblock {\em ACM Trans Math Softw}, 41(3):21:1--21:12, June 2015.
\newblock \href {http://dx.doi.org/10.1145/2699464}
  {\path{doi:10.1145/2699464}}.

\bibitem{Wolf2009}
Michael~M. Wolf, David {Perez-Garcia}, and Carlos Fernandez.
\newblock Measurements {{Incompatible}} in {{Quantum Theory Cannot Be Measured
  Jointly}} in {{Any Other No}}-{{Signaling Theory}}.
\newblock {\em Phys. Rev. Lett.}, 103(23):230402, December 2009.
\newblock \href {http://dx.doi.org/10.1103/PhysRevLett.103.230402}
  {\path{doi:10.1103/PhysRevLett.103.230402}}.

\bibitem{Yamashita2012}
Makoto Yamashita, Katsuki Fujisawa, Mituhiro Fukuda, Kazuhiro Kobayashi,
  Kazuhide Nakata, and Maho Nakata.
\newblock Latest {{Developments}} in the {{SDPA Family}} for {{Solving
  Large}}-{{Scale SDPs}}.
\newblock In Miguel~F. Anjos and Jean~B. Lasserre, editors, {\em Handbook on
  {{Semidefinite}}, {{Conic}} and {{Polynomial Optimization}}}, number 166 in
  International Series in Operations Research \& Management Science, pages
  687--713. {Springer US}, 2012.

\bibitem{Yang2014}
Tzyh~Haur Yang, Tam\'as V\'ertesi, Jean-Daniel Bancal, Valerio Scarani, and
  Miguel Navascu\'es.
\newblock Robust and {{Versatile Black}}-{{Box Certification}} of {{Quantum
  Devices}}.
\newblock {\em Phys. Rev. Lett.}, 113(4):040401, July 2014.
\newblock \href {http://dx.doi.org/10.1103/PhysRevLett.113.040401}
  {\path{doi:10.1103/PhysRevLett.113.040401}}.

\bibitem{Zhao2010}
X.~Zhao, D.~Sun, and K.~Toh.
\newblock A {{Newton}}-{{CG Augmented Lagrangian Method}} for {{Semidefinite
  Programming}}.
\newblock {\em SIAM J. Optim.}, 20(4):1737--1765, January 2010.
\newblock \href {http://dx.doi.org/10.1137/080718206}
  {\path{doi:10.1137/080718206}}.

\end{thebibliography}

\end{document}